\documentclass{article}

\usepackage{physics}
\usepackage{graphicx}
\usepackage{amsmath,amssymb}
\usepackage{amsthm}
\usepackage[english, footnote, draft, nomargin, marginclue]{fixme}
\fxusetheme{color}
\usepackage[colorlinks=true, pdfstartview=FitV, linkcolor=blue, citecolor=blue, urlcolor=blue]{hyperref}

\newtheorem{theorem}{Theorem}
\newtheorem{corollary}{Corollary}[theorem]

\begin{document}

\title{Anchored Network Users: Stochastic Evolutionary Dynamics of Cognitive Radio Network Selection 
}

\author{Ik Soo Lim, Peter Wittek
\thanks{The first author is with School of Computer Science, Bangor University, UK. e-mail: i.s.lim@bangor.ac.uk.
The second author is with ICFO-The Institute of Photonic Sciences, Barcelona Institute of
Science and Technology, 08860 Castelldefels (Barcelona), Spain.  
}
}

\markboth{Anchored Network Users
}
{}


\maketitle

\begin{abstract}
To solve the spectrum scarcity problem,
the cognitive radio technology involves licensed users and unlicensed users.
A fundamental issue for the network users is  whether it is better to act as a licensed user by using a primary network
or an unlicensed user by using a secondary network.
To model the network selection process by the users,
the deterministic replicator dynamics is often used, 
but in a less practical way that it requires each user to know global information on the network state
for reaching  a Nash equilibrium.
This paper addresses the network selection process
in a more practical way such that only noise-prone estimation of local information is required
and, yet, it obtains an efficient system performance.
\vspace{0.2cm}

\end{abstract}



\section{Introduction}
Cognitive radio (CR)  is a promising technology to solve the spectrum scarcity problem \cite{Liang:2011fk}. 
In CR networks, primary users (PUs) have licenses to operate in a certain spectrum band whereas secondary users (SUs) have no spectrum licenses and need to share  the spectrum holes left available by PUs without interfering with them.
In this paper,
we focus on a fundamental issue of
whether it is better for a CR user to act as a PU	
with guaranteed quality-of-service  at a higher price
or an SU
with degraded quality-of-service at a lower price.
Ref.\,\cite{Elias:2013qf} addresses this network selection problem
by the deterministic evolutionary dynamics
based on replicator equations,
assuming that each CR user
dynamically adjusts its network selection.
Although it would 
lead to the Nash equilibrium  
of network traffic yielding an efficient system performance,
the approach is less practical 
in the sense that
it assumes each of CR users to know the exact global information on the network state.
In this paper,
we address the network selection problem in a more realistic way, assuming that each CR user only needs to know error-prone local information.

\section{Network Models \& Equilibrium Computation}

As in \cite{Elias:2013qf},
we consider  a CR system  consisting of a primary network and a secondary network with a population of CR users,
where the secondary network  coexists with the primary one at the same location and on the same spectrum band.
Once the primary and secondary operators set the prices of network subscription,
each of the CR users dynamically chooses the network to use.
The wireless channel is modelled as an $M|M|1$ queue with the service rate (i.e.\,the maximum achievable transmission rate) $C$ and the arrival rate $\lambda$.
The cost $-\pi_{\scriptscriptstyle P}$ (or utility $\pi_{\scriptscriptstyle P}$) perceived by a PU is a combination of the service delay experienced in the network and the price to access this network,
\begin{equation}
-\pi_{\scriptscriptstyle P} = \frac{\alpha}{C -\lambda_{\scriptscriptstyle P}} +p_1 =\frac{\alpha}{C -\lambda x_{\scriptscriptstyle P}} +p_1
\label{eq_util_primary}
\end{equation}
where $\lambda_{\scriptscriptstyle P}$ denotes the overall transmission rate of PUs,
$\alpha$ a weighting parameter of delay with respect to the network subscription price $p_1$ charged by the primary network operator,
and
$x_{\scriptscriptstyle P} \in [0,1]$ the frequency or population share of PUs.
The cost $-\pi_{\scriptscriptstyle S}$ by an SU is 
\begin{equation}
-\pi_{\scriptscriptstyle S} = \frac{\alpha}{C -\lambda} +p_2
\label{eq_util_secondary}
\end{equation}
where $p_2$ denotes the price charged by the secondary network operator.
PUs and SUs experiencing  the same cost,
the equilibrium traffic $\lambda_{\scriptscriptstyle P}^*$ for the primary network is 
\begin{equation} 
\lambda_{\scriptscriptstyle P}^*
=\frac{\alpha \lambda -C(C -\lambda)(p_1 -p_2)}{\alpha -(C -\lambda)(p_1 -p_2)}.
\end{equation}
See Ref.\,\cite{Elias:2013qf} for the justification of the cost functions as well as the derivation of the prices and the equilibrium traffic.

\section{Critical Review of Applications of Replicator Dynamics}

Even if there exists  the equilibrium traffic that yields an efficient system performance,
it is a different matter whether CR users can reach the equilibrium.
For the latter,
Ref.\,\cite{Elias:2013qf} models 
the network selection process of CR users according to replicator dynamics, 
where users  individually adjust their selection based on the observed network state.

\subsection{Replicator Dynamics}

Originating from evolutionary biology, 
the replicator dynamics describes how the frequency of 
individuals using a strategy
in a population changes over time 
under the natural selection \cite{taylor1978evolutionary}.
Given a population of $n_i$ individuals using strategy $i \in \{1,\ldots,I\}$,
the replicator dynamics is described
with a set of ordinary differential equations
\begin{equation}
\frac{d x_i}{dt}
   =  \mathcal{K} x_i\left(\pi_i(\mathbf{x})  -\bar{\pi}(\mathbf{x})\right)
   \label{eq_replicator}
\end{equation}
where  
$x_i =n_i/\sum_{j=1}^I n_j$ is the frequency of individuals using strategy $i$,
$\mathcal{K}$ a constant, 
$\mathbf{x} =(x_1,\ldots,x_i,\ldots,x_I)$ the population state,
$\pi_i (\mathbf{x})$  the (expected) utility of strategy $i$,
and $\bar{\pi} (\mathbf{x}) =\sum_{i=1}^I x_i \pi_i(\mathbf{x})$
 the population mean of utility.
According to Eq.\,\ref{eq_replicator},
the frequency $x_i$ increases when its utility is larger than the population mean 
and it decreases when its utility is lower than the mean.

The replicator dynamics leads the population of individuals to a Nash equilibrium.
Because of this favourable feature,
the replicator equations have been widely applied to describe individuals to adaptively adjust their strategies over time
and reach a Nash equilibrium in problems related to CR networks
\cite{Elias:2013qf, niyato2009dynamics, 
Chen:2013cr}.
In this setting,
strategy $i$ is analogous to a strategy of, say, choosing network $i$. 
These applications interpret the replicator equations
as the description of how each of individuals should behave.
According to this interpretation,
however,
each individual is required to know some of the global information,
which makes it less practical.
In the network selection problem,
for instance,
Ref.\,\cite{Elias:2013qf}
assumes that CR users need to know the global information such as  $x_{\scriptscriptstyle P}$  and  $x_{\scriptscriptstyle S}$ (the frequency of PUs and SUs, respectively)
in order to select their networks.

\subsection{From Individual Behaviours to Population Dynamics}

The replicator equations Eq.\,\ref{eq_replicator} describe the dynamics at a population level,
but not necessarily specify how each individual should choose a pure strategy \cite{Moon:2016ly}.
In the original setting of evolutionary biology,
the replicator equations describe the population dynamics 
arising from a set of individuals replicating themselves by reproduction in a way proportional to their utilities,
not requiring any global information  \cite{taylor1978evolutionary}.
Other than reproduction,
social learning or imitation of pure strategies
can also yield the replicator population dynamics  \cite{schlag1998why}.	
The relation between individual behaviour and population dynamics can be more explicitly represented in the following form
\begin{equation}
\frac{d x_i}{dt} = \sum_{j=1}^I x_j \rho_{ji}(\mathbf{x}) -x_i\sum_{j=1}^I \rho_{ij}(\mathbf{x})
\label{eq_Protocol_Population}
\end{equation}
where $\rho_{ij}$ denotes a revision protocol or a conditional switch rate that 
describes when and how an individual in the population decide to switch strategy $i$ to $j$ \cite{Sandholm:2003kl}.
The first summation captures the in-flow of individuals switching to strategy $i$ 
and the second one, the out-flow of those switching from $i$ to other strategies.
	
The revision protocol of pairwise proportional imitation
$\rho_{ij}(\mathbf{x}) = x_j [\pi_j(\mathbf{x}) - \pi_i(\mathbf{x})]_+$
(where $[z]_+ =z$  if $z>0$ and $[z]_+ =0$ if $z \le0$)
yields the replicator population dynamics of Eq.\,\ref{eq_replicator},
which can be easily shown by plugging $\rho_{ij}(\mathbf{x})$ into Eq.\,\ref{eq_Protocol_Population}
 \cite{schlag1998why}.
 Note that an individual does need to know the population share $x_j$ for the pairwise imitation.
 This term in the protocol merely accounts for an individual to randomly choose an opponent,
 who of strategy $j$ is selected with probability $x_j$.
The individual imitates the strategy of the opponent only if the opponent's utility is higher than his own, doing so with probability proportional to the utility difference.        
The pairwise imitation drives the system to a Nash equilibrium without any need of global information.
The  imitation protocol and a variant of it  have been recently applied to network-related problems,
in order to reach a Nash equilibrium of the system-wide optimum   in distributed manners \cite{Iellamo:2013fk, Chen:2015dq}.
If applied to the network selection problem, thus,
the imitation protocol would yield the Nash equilibrium of network traffic in a more practical manner than the approach of Ref.\,\cite{Elias:2013qf} does.

\section{Imitation-based Network Selection}

\subsection{Markov Chains} 

We use a Markov chain
to model a  population of CR users conducting
the imitation-based network selection.	
The finite state space of the Markov chain is
$\mathcal{F}^N =\{k : 0 \le k \le N\}$
where $k$ is an integer-valued random variable denoting the number of PUs among $N$ network users; 
there are $N-k$ of SUs.
The arrival rate of PUs is
$\lambda_{\scriptscriptstyle P}(k) =\lambda k /N$.
Since there are only two strategies,
the stochastic evolution can be described by a birth-death process 
 on the one-dimensional finite state space $\mathcal{F}^N$ \cite{KARLIN:1975fc}.
 In each stochastic event, the state variable $k$ can either remain unchanged or move to $k+1$ or $k-1$.
With transitions only occurring between adjacent states,
the transition probabilities are 
\begin{align}
T_k^+ & =  \frac{N-k}{N}  \frac{k}{N-1}
q_{\scriptscriptstyle S \rightarrow  P}^k,\\
T_k^- & = \frac{k}{N} \frac{N-k}{N-1} 
q_{\scriptscriptstyle P \rightarrow S}^k,\\
T_k^0 & = 1 -T_k^+ -T_k^-  
\end{align}
where $T_k^+$ denotes the probability of a transition from state $k$ to $k +1$,
$T_k^-$  from $k$ to $k -1$,
$T_k^0$ remaining in  $k$,
$q_{\scriptscriptstyle \scriptscriptstyle S \rightarrow   P}^k$ the probability 
of an SU imitating 
a given PU (i.e.\,switching to be a PU) when the total number of PUs is $k$, 	 	
$q_{\scriptscriptstyle P \rightarrow S}^k$ the probability of a PU imitating a given SU.
We assumes a nondecreasing
function 
$q(z)$ 
for the imitation probabilities
$q_{\scriptscriptstyle S \rightarrow  P}^k =q\left(\pi_{\scriptscriptstyle P}^k -\pi_{\scriptscriptstyle S}^k\right)$ and
$q_{\scriptscriptstyle \scriptscriptstyle P \rightarrow S}^k =q\left(\pi_{\scriptscriptstyle S}^k -\pi_{\scriptscriptstyle P}^k\right)$.

\subsection{Noise-free Imitation}
For the	
noise-free imitation,
we have the imitation probability $q(z) = 0$ for $z \le 0$
and $q(z)$ strictly increasing for $z >0$.
Under the revision protocol of the pairwise proportional imitation,
for instance,
we have $q_{\scriptscriptstyle S \rightarrow  P}^k =q\left(\pi_{\scriptscriptstyle P}^k -\pi_{\scriptscriptstyle S}^k\right) = [\pi_{\scriptscriptstyle P}^k -\pi_{\scriptscriptstyle S}^k]_+$
and $q_{\scriptscriptstyle P \rightarrow S}^k =q\left(\pi_{\scriptscriptstyle S}^k -\pi_{\scriptscriptstyle P}^k\right) =[\pi_{\scriptscriptstyle S}^k -\pi_{\scriptscriptstyle P}^k]_+$.
Let us define 
 $k^*\equiv \lceil N x_{\scriptscriptstyle P}^*\rceil = \lceil{N \lambda_{\scriptscriptstyle P}^* /\lambda}\rceil$
where $x_{\scriptscriptstyle P}^*$ denotes the equilibrium point of the replicator equation
$d x_{\scriptscriptstyle P} /dt
   =  \mathcal{K} x_{\scriptscriptstyle P}\left(\pi_{\scriptscriptstyle P}(\mathbf{x})  -\bar{\pi}(\mathbf{x})\right)
$,
$x_{\scriptscriptstyle P}^* = \lambda_{\scriptscriptstyle P}^* /\lambda$
and $\lceil z \rceil$ the least integer greater than or equal to $z$.
Since $\pi_{\scriptscriptstyle P} > \pi_{\scriptscriptstyle S}$ for $x_{\scriptscriptstyle P} < x_{\scriptscriptstyle P}^*$, 
$\pi_{\scriptscriptstyle P} = \pi_{\scriptscriptstyle S}$ for $x_{\scriptscriptstyle P} = x_{\scriptscriptstyle P}^*$
and
$\pi_{\scriptscriptstyle P} < \pi_{\scriptscriptstyle S}$ for $x_{\scriptscriptstyle P} > x_{\scriptscriptstyle P}^*$
from Eq.\,\ref{eq_util_primary}
and \ref{eq_util_secondary},
we get $T_k^- =0$ for $k \le k^* -1$ and $T_k^+ =0$ for $k \ge k^*$,
assuming a non-boundary initial state (i.e.\,$k \ne 0, N$
at time $t=0$).
Thus,
we have the stationary probability distribution of the system
\begin{equation}
\psi_{k^* -1} =\frac{T^-_{k^*}}{T^+_{k^* -1} +T^-_{k^*}},
\psi_{k^*} =\frac{T^+_{k^* -1}}{T^+_{k^* -1} +T^-_{k^*}},\\
\psi_k =0 \mbox{ otherwise.} 
\label{eq_dist_free}
\end{equation}
The deterministic replicator dynamics well approximates the population dynamics
arising from the 
imitation-based decision process by individual network users
in the sense that the Nash equilibrium $x_{\scriptscriptstyle P}^*$ reached by the replicator dynamics well approximates the stationary distribution $\psi_k$
concentrated around 
$k=k^*  \approx  x_{\scriptscriptstyle P}^* N
$.

\subsection{Noisy Imitation}

Although the imitation protocol
relaxes the requirement of the global information,
it still suffers from an unrealistic assumption.
It assumes 	
that a user should never imitate an opponent user  of a lower utility.
In practice, 
it is difficult to strictly meet this assumption of `noise-free' imitation due to various reasons.
In the network selection problem,
a user needs to observe and estimate the expected service delay in the network,
which in general deviates from the ground truth of the expected delay.
Being self-interested, an opponent user may deliberately inform of inaccurate utility information.
Thus, it is more realistic to assume the `noisy' imitation such that a user could imitate an opponent of lower utility,
due to a decision-making based on the error-prone estimations.
	
For the noisy imitation, we assume that the imitation probability 
$q(z)$ is strictly increasing.
The key difference from that of the noise-free imitation is
$q(z) > 0$ even for $z \le 0$, reflecting the possible switch to the other network of a lower utility although the probability of such suboptimal behaviour is smaller than that of switching from a network of lower utility to a higher one.
Since we have $T_k^+ > 0$ and $T_k^- > 0$ for $k \in \{1,2,\ldots,N-1\}$,
the boundary states $k=0$ and $k=N$ are reachable from any other states $\{1,..,N-1\}$.
We also have $T_0^0 = T_N^0 =1$.
Therefore, a Markov chain for the noisy imitation is an absorbing Markov chain
with two absorbing states $k=0$ and $k=N$, 
which are the only stationary states.
All the other states $\{1,..,N-1\}$ are transient, including $k=k^* -1$ and $k=k^*$ 
that would correspond to the Nash equilibrium.
In other words,  regardless of the initial state,
the noisy imitation leads the system to end up with either 
 all-SUs ($k=0$) or all-PUs ($k=N$),
 driving the system away from the Nash equilibrium.
Failing to capture this stochastic effect,
the replicator dynamics is less than adequate 
to model the network selection process in the noise-prone realistic situations.

\section{Network Selection with Anchored Users}

The absorbing states of all-PUs and all-SUs
are not good for the network operators nor CR users.
An operator with no CR users of its network collects no income
while CR users suffer from the utility lower than the 
one that would be obtained at the Nash equilibrium.
Thus, there is a clear need to prohibit the system from being absorbed in any of the suboptimal states.

\subsection{Noisy Imitation with Anchored Users}	
The boundary states $k=0$ and $k=N$ are the absorbing states under the noisy imitation
because there is no individual of a different strategy available for imitation and hence no change, 
once in one of the two states.
However,
the absorbing states could be avoided if some individuals behave irrespective of their utility \cite{Binmore:1997ly,Sandholm:2012ly}.
In the context of the network selection,
for instance,
if at least one user for each network
never switches the network (irrespective of the utility)
 and is always available for imitation by other users,
 then there would be no absorbing state.
However,
it is a strong assumption that anyone among self-interested users should act like this, which would be a kind of an altruistic act.  
On the other hand,
it would be rational for a self-interested network operator to set up `puppet' users 
that are anchored to the network at the operator's own cost
since it ensures avoiding the extinction of its genuine users;
the anchored users are always available for imitation by genuine users of the other network.

With the anchored users in place,
the transition probabilities are 
\begin{align}
T_k^+ & =  \frac{N-k}{N}  \frac{k+A_{\scriptscriptstyle P}}{N-1+A_{\scriptscriptstyle P} +A_{\scriptscriptstyle S}}q\left(\pi_{\scriptscriptstyle P}^k -\pi_{\scriptscriptstyle S}^k\right), \\
T_k^- & = \frac{k}{N} \frac{N-k+A_{\scriptscriptstyle S}}{N-1+A_{\scriptscriptstyle P} +A_{\scriptscriptstyle S}} q\left(\pi_{\scriptscriptstyle S}^k -\pi_{\scriptscriptstyle P}^k\right), \\
T_k^0 & = 1 -T_k^+ -T_k^-  
\label{eq_transition_Prob} 
\end{align}
where $A_{\scriptscriptstyle P}$ and $A_{\scriptscriptstyle S}$ denote the number of anchored PUs and SUs, respectively.
Note that $N$ and $k$ count only genuine users, but not  anchored (puppet) users.
Since $T_k^+ >0$  for $k \in \{0,1,\ldots,N-1\}$ and $T_k^- >0$ for $k \in \{1,\ldots,N\}$,
it is possible to go from every state to every state.
With the anchored users,
hence, the Markov chain for the noisy imitation is not absorbing anymore,
but it is irreducible. 
An irreducible Markov chain yields a unique stationary distribution 
that indicates the likelihood of finding the population in any particular state in the long-run.
Although it can be generally obtained as the left eigenvector of the transition matrix with eigenvalue one,
the stationary distribution $\psi_k$ for a birth-death process with two strategies can be explicitly represented by	
\begin{equation}
	\psi_k =\psi_0\Pi^k_{n=1} T^+_{n-1} /T^-_n
	\label{eq_Stationary_dist}
\end{equation}
where $\psi_0$ is determined by $\sum_{k=0}^N \psi_k =1$
\cite{gardiner2004handbook}.

\subsection{Stationary Distribution and Nash Equilibrium}

With the inclusion of the anchored users, we can not only remove the sub-optimal absorbing states
but also establish a link between the system states of the noise-free and noisy imitation protocols.	
For $A_{\scriptscriptstyle P} = A_{\scriptscriptstyle S} =1$,
we show that the peak of the stationary distribution well corresponds to
the Nash equilibrium 
that the replicator population dynamics arising from  the noise-free imitation would drive the system towards.

\begin{theorem}
For $A_{\scriptscriptstyle P} = A_{\scriptscriptstyle S} =1$,
we have $\arg \max_{k\in \{0,\ldots,N\}}\psi_k   = k^*-1 \mbox{ or } k^*$
where $k^*=\lceil N x_{\scriptscriptstyle P}^*\rceil = \lceil{N \lambda_{\scriptscriptstyle P}^* /\lambda}\rceil$
\end{theorem}

\begin{proof}
Note that 
$\psi_k = T^+_{k-1} /T^-_k \psi_{k-1}
$
since
$\psi_k  = \psi_0 \Pi^k_{n=1} T^+_{n-1} /T^-_n 
= \psi_0 T^+_{k-1} /T^-_k \Pi^{k-1}_{n=1} T^+_{n-1} /T^-_n
= T^+_{k-1} /T^-_k \psi_{k-1}
$.
Let $\pi_{\scriptscriptstyle P}^k \equiv \pi_{\scriptscriptstyle P}(k) = \alpha /\left(C -\lambda_{\scriptscriptstyle P}\right) +p_1 =\alpha /\left(C -\lambda x_{\scriptscriptstyle P}\right) +p_1 =\alpha /\left(C -\lambda k/N\right) +p_1$
and $\pi_{\scriptscriptstyle S}^k \equiv \pi_{\scriptscriptstyle S}(k) = \pi_{\scriptscriptstyle S} =\alpha /\left(C -\lambda\right) +p_2$.
Note that $\pi_{\scriptscriptstyle P}^k > \pi_{\scriptscriptstyle S}^k$ for $k/N < \lambda_{\scriptscriptstyle P}^*/\lambda_{\scriptscriptstyle P}$,
$\pi_{\scriptscriptstyle P}^k = \pi_{\scriptscriptstyle S}^k$ for $k/N = \lambda_{\scriptscriptstyle P}^*/\lambda_{\scriptscriptstyle P}$,
and $\pi_{\scriptscriptstyle P}^k < \pi_{\scriptscriptstyle S}^k$ for $k/N > \lambda_{\scriptscriptstyle P}^*/\lambda_{\scriptscriptstyle P}$.
For $A_{\scriptscriptstyle P} =A_{\scriptscriptstyle S} =1$,
we have 
$T^+_{k-1} /T^-_k         = q\left(\pi_{\scriptscriptstyle P}^{k-1} -\pi_{\scriptscriptstyle S}^{k-1}\right) /q\left(\pi_{\scriptscriptstyle S}^k -\pi_{\scriptscriptstyle P}^k\right) $.
               
For an integer $k \le k^* -1$, we have $\pi_{\scriptscriptstyle P}^k > \pi_{\scriptscriptstyle S}^k$.
Since $q(\cdot)$ is strictly increasing  as well as $\pi_{\scriptscriptstyle P}^{k-1} - \pi_{\scriptscriptstyle S}^{k-1} > 0$ and $\pi_{\scriptscriptstyle S}^k - \pi_{\scriptscriptstyle P}^k < 0$,
we have $T^+_{k-1} /T^-_k         = q\left(\pi_{\scriptscriptstyle P}^{k-1} -\pi_{\scriptscriptstyle S}^{k-1}\right) /q\left(\pi_{\scriptscriptstyle S}^k -\pi_{\scriptscriptstyle P}^k\right) >1$,
yielding
$\psi_k = T^+_{k-1} /T^-_k \psi_{k-1}
> \psi_{k-1} 
$.
In other words,
$\psi_k$ increases as $k$  increases as far as $k \le k^* -1$ and,
hence, $\arg \max_{k\in \{0,\ldots,k^* -1\}}\psi_k  = k^* -1$.

For $k \ge k^* +1$, we have $\pi_{\scriptscriptstyle P}^k < \pi_{\scriptscriptstyle S}^k$.
Since $\pi_{\scriptscriptstyle P}^{k-1} -\pi_{\scriptscriptstyle S}^{k-1} \le 0$  and $\pi_{\scriptscriptstyle S}^k - \pi_{\scriptscriptstyle P}^k >0$,
we have $T^+_{k-1} /T^-_k         = q\left(\pi_{\scriptscriptstyle P}^{k-1} -\pi_{\scriptscriptstyle S}^{k-1}\right) /q\left(\pi_{\scriptscriptstyle S}^k -\pi_{\scriptscriptstyle P}^k\right) <1$,
yielding
$\psi_k= T^+_{k-1} /T^-_k \psi_{k-1}
< \psi_{k-1}
$.
Note that $\psi_{k^*+1}= T^+_{k^*-1} /T^-_{k^*} \psi_{k^*}
< \psi_{k^*}$ holds as well.
In other words,
$\psi_k$ decreases with $k$  for $k \ge k^* $
and, hence, $\arg \max_{k\in \{k^*,\ldots,N\}}\psi_k  = k^*$.

In conclusion,
$\arg \max_{k\in \{0,\ldots,N\}}\psi_k   = k^*-1 \mbox{ or } k^*$.
\end{proof}

\begin{corollary}
For $A_{\scriptscriptstyle P} = A_{\scriptscriptstyle S} =1$,
we have
\[
    \arg \max_{k\in \{0,\ldots,N\}}\psi_{k} = 
\begin{cases}
    k^*,& \text{if } |\triangle \pi_{\scriptscriptstyle P,S}^{k^* -1}| > |\triangle \pi_{\scriptscriptstyle P,S}^{k^*}| \\
    \{k^*-1, k^*\},              & \text{if } |\triangle \pi_{\scriptscriptstyle P,S}^{k^* -1}| = |\triangle \pi_{\scriptscriptstyle P,S}^{k^*}| \\
    k^*-1,              & \text{if } |\triangle \pi_{\scriptscriptstyle P,S}^{k^* -1}| < |\triangle \pi_{\scriptscriptstyle P,S}^{k^*}|
\end{cases}
\]
where $| \triangle \pi_{\scriptscriptstyle P,S}^{k^* -1} | \equiv | \pi_{\scriptscriptstyle P}^{k^* -1} -\pi_{\scriptscriptstyle S}^{k^* -1} |$ 
and
$| \triangle \pi_{\scriptscriptstyle P,S}^{k^*} | \equiv | \pi_{\scriptscriptstyle P}^{k^*} -\pi_{\scriptscriptstyle S}^{k^*} |$.

\end{corollary}

\begin{proof}
Since $\pi_{\scriptscriptstyle P}^k > \pi_{\scriptscriptstyle S}^k$ for $k < k^*$ and $\pi_{\scriptscriptstyle P}^{k^*} -\pi_{\scriptscriptstyle S}^{k^*} \le 0$,
we have $\pi_{\scriptscriptstyle P}^{k^*-1} -\pi_{\scriptscriptstyle S}^{k^*-1} > 0$ and $\pi_{\scriptscriptstyle S}^{k^*} - \pi_{\scriptscriptstyle P}^{k^*} \ge 0$.

For $\pi_{\scriptscriptstyle P}^{k^*-1} -\pi_{\scriptscriptstyle S}^{k^*-1} > \pi_{\scriptscriptstyle S}^{k^*} - \pi_{\scriptscriptstyle P}^{k^*}$,
we have $T^+_{k^*-1} /T^-_{k^*}         = q\left(\pi_{\scriptscriptstyle P}^{k^*-1} -\pi_{\scriptscriptstyle S}^{k^*-1}\right) /q\left(\pi_{\scriptscriptstyle S}^{k^*} -\pi_{\scriptscriptstyle P}^{k^*}\right) >1$,
yielding $\psi_{k^*} = T^+_{k^*-1} /T^-_{k^*} \psi_{k^*-1} > \psi_{k^*-1}$ 
and, thus,
$\arg \max_{k\in \{0,\ldots,N\}}\psi_{k}   =k^*$.

For $\pi_{\scriptscriptstyle P}^{k^*-1} -\pi_{\scriptscriptstyle S}^{k^*-1} = \pi_{\scriptscriptstyle S}^{k^*} - \pi_{\scriptscriptstyle P}^{k^*}$, we have $T^+_{k^*-1} /T^-_{k^*}         = q\left(\pi_{\scriptscriptstyle P}^{k^*-1} -\pi_{\scriptscriptstyle S}^{k^*-1}\right) /q\left(\pi_{\scriptscriptstyle S}^{k^*} -\pi_{\scriptscriptstyle P}^{k^*}\right)= 1$,
yielding $\psi_{k^*} 
= T^+_{k^*-1} /T^-_{k^*}\psi_{k^*-1} =\psi_{k^*-1}$  and, hence, 
$\arg \max_{k\in \{0,\ldots,N\}}\psi_k   =\{k^*-1, k^*\}$.

For $\pi_{\scriptscriptstyle P}^{k^*-1} -\pi_{\scriptscriptstyle S}^{k^*-1} < \pi_{\scriptscriptstyle S}^{k^*} - \pi_{\scriptscriptstyle P}^{k^*}$,
we have $T^+_{k^*-1} / T^-_{k^*}       = q\left(\pi_{\scriptscriptstyle P}^{k^*-1} -\pi_{\scriptscriptstyle S}^{k^*-1}\right) /q\left(\pi_{\scriptscriptstyle S}^{k^*} -\pi_{\scriptscriptstyle P}^{k^*}\right) < 1$,
yielding $\psi_{k^*} =T^+_{k^*-1} /T^-_{k^*} \psi_{k^*-1}
< \psi_{k^*-1}$ and, thus, 
$\arg \max_{k\in \{0,\ldots,N\}}\psi_k   =k^*-1$.
\end{proof}

\section{Social Welfare}
We need to measure the system efficiency under the noisy imitation without/with anchored users
as well as the noise-free imitation.
	
\subsection{Price of Anarchy}

The Price of Anarchy (PoA) is one of the most popular performance metrics, quantifying the loss of efficiency as the ratio between the cost of the worst stable outcome  and the cost of the optimal outcome \cite{Koutsoupias:2009fk}.
As in Ref.\,\cite{Elias:2013qf},
we define
social welfare    
$S(x_{\scriptscriptstyle P}) 
=\lambda\left[x_{\scriptscriptstyle P} /\left(C -\lambda x_{\scriptscriptstyle P}\right) +\left(1 -x_{\scriptscriptstyle P}\right) /\left(C -\lambda\right)\right]
$
as the total delay experienced by PUs and SUs
and
 \begin{equation}
\mbox{PoA} = S(x_{\scriptscriptstyle P}^*)/S
_{\min}
 \end{equation} 
where   
$S(x_{\scriptscriptstyle P}^*)$ denotes the total delay experienced at the stable equilibrium point $x_{\scriptscriptstyle P}^*$
and
$S_{\min} =2 \left(\sqrt{C/(C-\lambda)} -1\right)$ is the social optimum of the total delay,
which is obtained at $x_{\scriptscriptstyle P}$ where $d S / d x_{\scriptscriptstyle P} =0$ \cite{Elias:2013qf}.
For the noisy imitation without anchored users, 
we have 
\begin{equation}
\mbox{PoA}  =\frac{\lambda}{C -\lambda} \frac{1}{S_{\min}}
=\frac{\lambda}{2\sqrt{C -\lambda}\left(\sqrt{C} -\sqrt{C -\lambda}\right)}
\end{equation}
because the two absorbing states ($k=0,N$) corresponding to $x_{\scriptscriptstyle P}^*=0 \mbox{ and } 1$ are the only stable states 
as well as
$S(0) =S(1) =\lambda /\left(C -\lambda\right)$,

\subsection{Expected Price of Anarchy}

Since the noisy imitation with the anchored users yields an irreducible Markov chain that does not have any equilibrium point,
PoA is not suitable as a performance metric.
Because  an irreducible Markov chain  yields a unique stationary distribution,
we instead use Stationary Expected Social Welfare  
$\mathbb{E}_{\delta}[S] = \sum_{k=0}^N S\left(x_{\scriptscriptstyle P}(k)\right)\psi_i
$
where $S\left(x_{\scriptscriptstyle P}(k)\right)$ denotes the social welfare when the total number of PUs is $k$
and  $x_{\scriptscriptstyle P}(k) =k/N$  \cite{Auletta:2013uq}.
Analogous to PoA,
we define the Expected PoA as
\begin{equation}
\mbox{PoA}_{E} =\frac{\mathbb{E}_{\delta}[S]}{S_{\min}}.
\end{equation}
Even for 
the noise-free imitation,
$\mbox{PoA}_E$ is better suited than PoA
because the long-run state 
is a stationary probability distribution $\psi_k$ (Eq.\,\ref{eq_dist_free})
rather than an equilibrium point
due to the discrete nature of the system.

\section{Results}

We set $C=100$, and $\alpha=1$  as in Ref.\,\cite{Elias:2013qf}.

\subsection{Noise-free vs.\,Noisy Imitation}

Fig.\,\ref{fig_noise-free_imitation}\,(a) shows the long-run state  of a Markov chain 
under the noise-free  imitation protocol.
It is well approximated by the Nash equilibrium
that the replicator dynamics predicts.
Fig.\,\ref{fig_noise-free_imitation}\,(b) shows the $\mbox{PoA}_E$ of the system.
We have an efficient system performance,
yielding only a small loss of efficiency with respect to the social optimum (i.e.\,$\mbox{PoA}_E \approx 1$)
unless the network traffic $\lambda$ is too close to the capacity $C$ and the population size $N$ is too small.
\begin{figure}[t] 
\centering
      \includegraphics[height=0.25\textwidth]{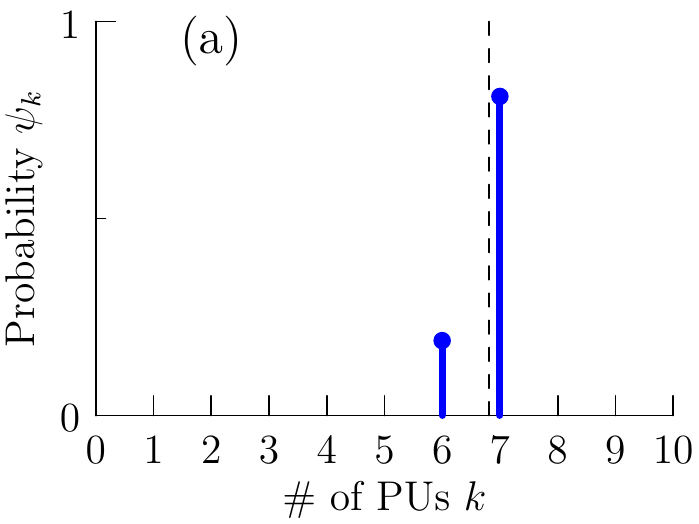}   
\includegraphics[height=0.25\textwidth]{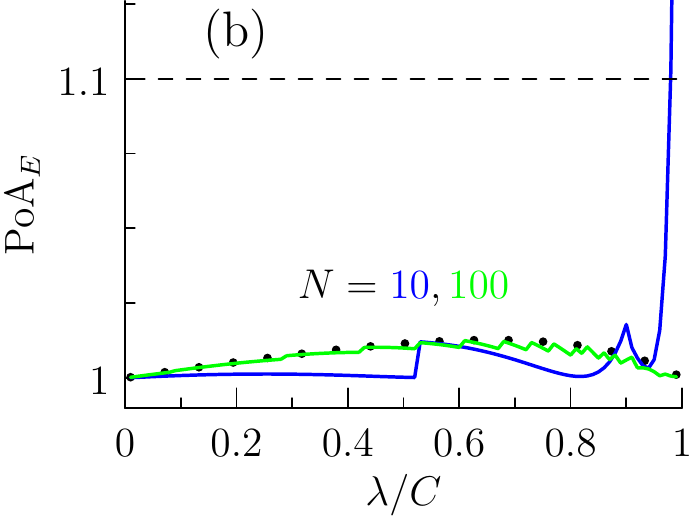}             
\caption{
The system  under  the noise-free imitation.
(\textbf{a}) 
The stationary probability distribution $\psi_k$ of the number of PUs
with an interior initial condition  $k \in \{1, \ldots,N-1\}$ where $N=10$ and $\lambda=30$.
The distribution is concentrated only on $k =k^* -1$ and $k =k^*$ 
where $k^* =7$.
The dashed vertical  line corresponds to the Nash equilibrium $N x^*_P (\approx 6.8)$ 
predicted by the deterministic replicator equation.
(\textbf{b})\,$\mbox{PoA}_E$ under various traffic $\lambda$ with $N=10$ and $100$.
The dashed horizontal line indicates $ \mbox{PoA}_E=1.1$, corresponding to a loss of efficiency of 10\% with respect to the social optimum.
The dotted curve indicates PoA of the Nash equilibrium,
which tends to be higher than $ \mbox{PoA}_E$ 
unless $\lambda/C \approx 1$ (where $C=100$) and $N$ is too small.
}
\label{fig_noise-free_imitation}
\end{figure} 
Fig.\,\ref{fig_noisy_imitation}\,(a) shows the outcome under the noisy imitation without anchored users.
The noisy imitation yields absorbing states of either all-SUs 
or all-PUs,
which the deterministic replicator dynamics is inadequate to capture.
Fig.\,\ref{fig_noisy_imitation}\,(b) shows the corresponding PoA
that reveals significant loss of efficiency  in a wide range of the network traffic.

\begin{figure}[t]
\centering
\includegraphics[height=0.25\textwidth]{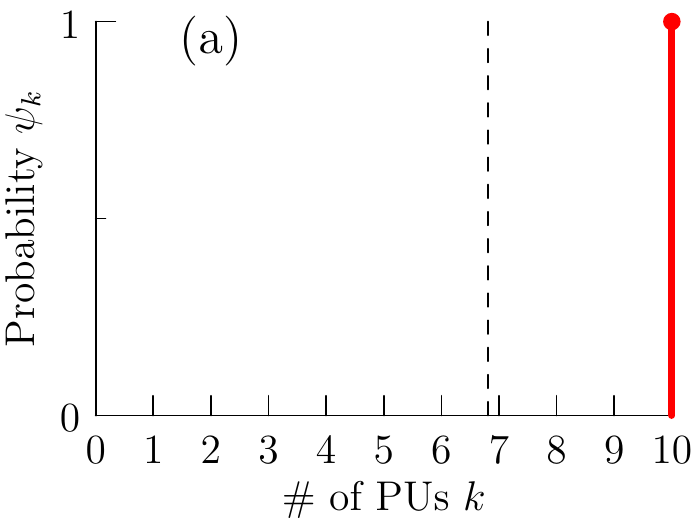}     
\includegraphics[height=0.25\textwidth]{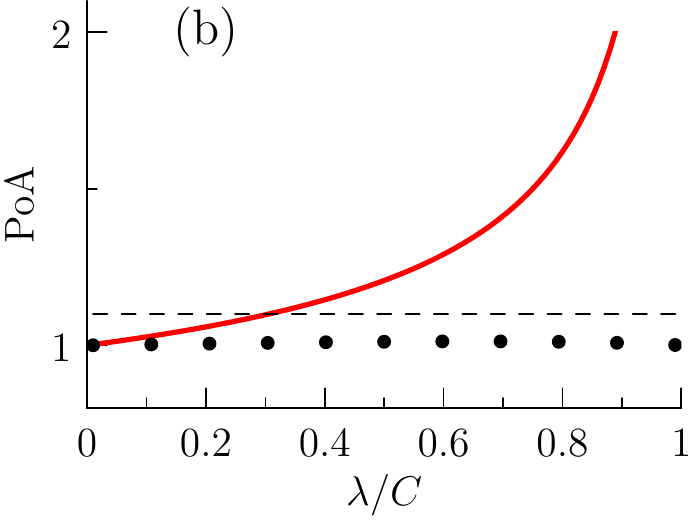}        
\caption{
The system under the noisy imitation.
(\textbf{a}) 
Regardless of the initial state,
the system ends up with one of the absorbing states, all-SUs ($k=0$) and all-PUs ($k=N$) where $N=10$ and $\lambda=30$;
only the case of all-PUs is shown.
(\textbf{b})\,PoA at the absorbing state. 
The loss of efficiency is significant  in a wide range of the network traffic,
e.g.\,PoA $>$ 1.1 for $\lambda/C > 0.35$. 
}
\label{fig_noisy_imitation} 
\end{figure}

\subsection{Noisy Imitation with Anchored Users}
We use Fermi function $q(z) =\left[1 +\exp(-\beta z)\right]^{-1} 
$
where $\beta$ controls 
the (inverse) level of noise,
which well captures the noisy imitation \cite{traulsen2006stochastic}.
Fig.\,\ref{fig_distribution} shows the stationary distributions due to an anchored user set up by each of the two network operators
as well as 
$\mbox{PoA}_E$.	
The peak of each distribution well corresponds to the Nash equilibrium 
predicted by the replicator population dynamics,
being within the distance $1/N$ of the Nash equilibrium.
The inclusion of the anchored users significantly improves the system efficiency under the noisy imitation
(i.e.\,$\mbox{PoA}_E < 1.1$) for a wider range of the traffic
than that
without anchored users. 
The system performance $\mbox{PoA}_E$ improves and converges towards PoA
of the Nash equilibrium  
as 
the level of noise drops 
(Fig.\,\ref{fig_distribution}\,(a))
or the population size increases
(Fig.\,\ref{fig_distribution}\,(b)).
Fig.\,\ref{fig_distribution}\,(b) also shows that Gaussian distributions  centred at the Nash equilibrium are good approximations of the stationary distributions,
the distributions being concentrated near the Nash equilibrium.
Even if each user can behave suboptimally,
the stationary distribution
converges toward that of an efficient system performance corresponding to the Nash equilibrium.

\begin{figure}[t]
\centering
\includegraphics[height=0.24\textwidth]{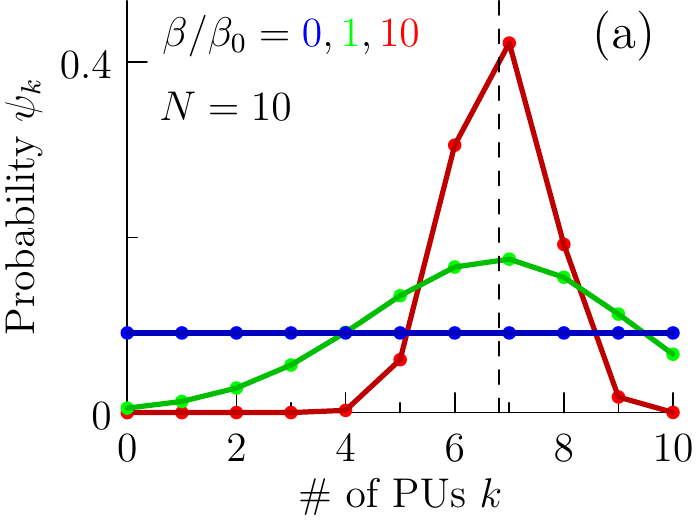}         
\includegraphics[height=0.24\textwidth]{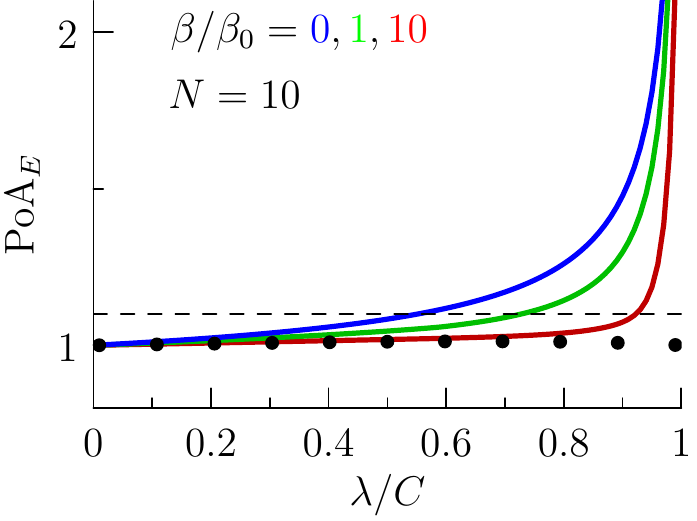}  
\vspace{0.2cm}\\
\includegraphics[height=0.25\textwidth]{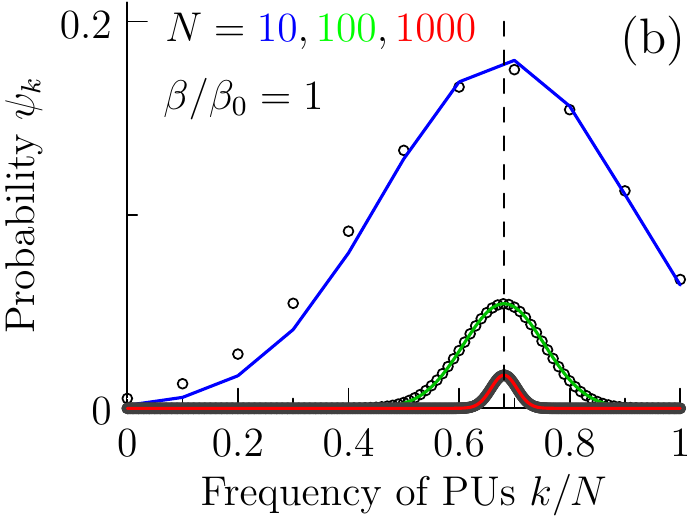}            
\includegraphics[height=0.25\textwidth]{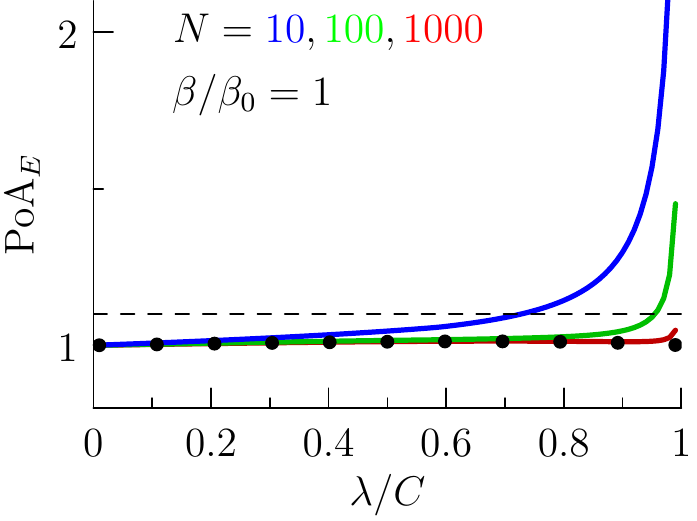}                
\caption{
The system under the noisy imitation with two anchored users; one per operator, $A_{\scriptscriptstyle P} =A_{\scriptscriptstyle S} =1$. 
(\textbf{a}) 
Stationary distributions and  $\mbox{PoA}_E$
at various levels of (inverse) noises $\beta / \beta_0 = 0, 1 \mbox{ and } 10$
where $\beta_0 = \max_k | \pi_{\scriptscriptstyle P}^k -\pi_{\scriptscriptstyle S}^k |$.
$N=10$ and $\lambda=30$.
The less noise (i.e.\,$\beta / \beta_0 \rightarrow \infty$), the narrower spreading of the distribution
and the better system performance.
(\textbf{b}) 
Stationary distributions and $\mbox{PoA}_E$
at various population sizes $N=10, 100$, and $1000$.
 $\beta / \beta_0 = 1$.
The larger population ($N \rightarrow \infty$), the narrower (relative) spreading of the distribution
and the better system performance.
} 
\label{fig_distribution}  
\end{figure}

\section{Conclusions}
For the network selection game between  primary and secondary networks,
we show that requiring only local information,
the noise-free imitation  among cognitive radio users 
drives the system to the state well approximated by the Nash equilibrium of
the replicator population dynamics
and yields an efficient system performance.
In more realistic situations,
however,
the imitation process becomes noisy
and it drives the system away from the Nash equilibrium to the state of either
all-primary users or all-secondary users, resulting in a sub-optimal system performance.
To overcome the sub-optimality of the noisy imitation, 
we introduce the notion of anchored network users to be set up by 
the self-interested network operators,
which yields a stationary distribution peaked at the Nash equilibrium.
It significantly  improves the system performance, which converges towards that of the Nash equilibrium.

%
%
%
%
%
%
%

\bibliographystyle{plain}
\bibliography{finite_population}


\end{document}